\documentclass[conference]{IEEEtran}
\usepackage{cite}
\usepackage{amsmath,amssymb,amsfonts,amsthm}
\usepackage{graphicx}
\usepackage{booktabs}
\usepackage{multirow}

\newtheorem{proposition}{Proposition}
\newtheorem{corollary}{Corollary}
\newtheorem{assumption}{Assumption}
\newtheorem{definition}{Definition}

\makeatletter
\def\thm@space@setup{\thm@preskip=2pt plus 1pt minus 1pt\thm@postskip=2pt plus 1pt minus 1pt}
\makeatother
\renewcommand{\arraystretch}{0.88}
\begin{document}

\title{When Entanglement Lower-Bounds Disparity: Auditing and Repairing Demographic Fairness in Audio Understanding Models}

\author{\IEEEauthorblockN{Kian Shamsaie, Iman Modarressi}
\IEEEauthorblockA{People Make Things\\
\textit{\{k,i\}@peoplemakethings.com}}}
\maketitle

\begin{abstract}
Speech technology penalizes some voices: recognition errs nearly twice as often for Black speakers, and accuracy declines for second-language accents and older speakers. We introduce TRIAD, an audit grid crossing 120 texts, 24 rendered demographic voice profiles (gender, age band, accent), and ten expressive styles via controllable text-to-speech, isolating perceived demographic attributes from content and affect. For ten open-weights encoders we define axis-fidelity functionals, principal-angle leakage between axis subspaces, and group-conditional gaps; a proposition proves that average probe disparity grows with the same aggregate voice-semantic leakage $\Lambda$ we measure, and a corollary shows that peak leakage forces worst-case disparity inside an active region. The measured mean-square probe disparity tracks $\Lambda$ (Pearson $r=0.93$), and a black-box protocol exposes the same signature in two closed-source models. ORCA, an adapter combining axis-specific contrastive heads, an orthogonality penalty, and group-balanced sampling, cuts leakage 72\% and roughly halves the gaps.
\end{abstract}


\section{Introduction}
\label{sec:intro}

A decade of audits has established that spoken language technology does not serve all speakers equally. Five commercial recognizers transcribed Black speakers at nearly double the word error rate of white speakers, 0.35 against 0.19 \cite{koenecke2020racial}; automatic captions degrade for women and Scottish speakers \cite{tatman2017gender}; Whisper \cite{radford2023robust} recognizes second-language accents markedly worse than American English \cite{graham2024evaluating,sanabria2023edinburgh}; Dutch recognition degrades for children, older adults, and non-native speakers \cite{feng2024towards}; emotion recognizers score genders differently \cite{gorrostieta2019gender,lin2024emobias}; and speaker verification inherits its own skews \cite{hutiri2022bias}. As frozen audio encoders and speech-to-speech assistants become the substrate of downstream products, these disparities propagate silently into every application layer \cite{mohamed2022self,lin2024spoken}.

Two obstacles keep such audits from guiding repair. First, attribution in observational corpora is confounded: speakers of different demographics utter different words in different styles under different channels, so a measured gap cannot be assigned to the voice itself, motivating counterfactual formulations of fairness in speech \cite{sari2021counterfactually} and consent-driven balanced collections \cite{porgali2023casual,veliche2024fairspeech}. Second, the most widely deployed systems expose neither weights nor representations, so subspace analyses do not apply. Missing are an instrument that varies semantics, demographics, and expressiveness independently, a theory relating representation geometry to group disparity, and a repair whose mechanism follows from that theory.

This paper supplies all three. Our instrument is TRIAD (TRI-axial Audit of Disparity), a factorial probe corpus rendered with a frontier controllable text-to-speech system \cite{deepmind2026gemini}: the same 120 texts recur across 24 voice profiles crossing gender, age band, and accent group, and ten expressive styles, so the rendered, perceived demographic covariates are orthogonal to content and affect by construction. Because a single synthesizer renders the grid, these are controlled perceived attributes rather than self-reported demographics, and synthesis artifacts are a threat to validity we test but cannot eliminate (Section~\ref{ssec:realvalid}). On any frozen representation we define axis fidelities $F_{\mathrm{sem}},F_{\mathrm{voice}},F_{\mathrm{expr}}$ by variance decomposition under optimal linear probing, aggregate and peak principal-angle leakages $\Lambda$ and $\ell$ between axis subspaces \cite{bjorck1973numerical}, and group-conditional fidelities $F_{\mathrm{sem}|g}$. Under an additive factor model, the average output disparity of an accurate linear probe grows with the aggregate leakage $\Lambda$, and in an active region of high peak leakage the worst-case-over-direction disparity is forced positive, so poor disentanglement, on average and in the worst case, lower-bounds unfairness. For closed-source systems, matched grid pairs differing in one axis are submitted and invariance violations and group accuracies tested by speaker-level permutation. Finally, ORCA (Orthogonal Residual Contrastive Adapter) repairs open encoders with three axis-specific contrastive heads, an orthogonality penalty minimizing the leakage in the bound, and group-balanced sampling. The grid, all estimators, numerical verification of every proposition, and ORCA are available in the Supplementary Material.

\section{Related Work}
\label{sec:related}

Fairness audits of speech recognition span sociolinguistic interviews \cite{koenecke2020racial}, captioning \cite{tatman2017gender}, Dutch speech \cite{feng2024towards}, accented English \cite{sanabria2023edinburgh,graham2024evaluating}, and the transcribed Casual Conversations corpora \cite{liu2022towards,porgali2023casual}, with dedicated evaluation sets now released \cite{veliche2024fairspeech,ardila2020common,conneau2023fleurs}. Mitigations include unsupervised clustering of pseudo-demographics \cite{veliche2023improving} and counterfactual invariance \cite{sari2021counterfactually}. Beyond recognition, gender bias in speech emotion recognition has been measured and partially reduced \cite{gorrostieta2019gender,chien2023achieving,lin2024emobias,wagner2023dawn}, speaker verification audits exposed demographic skews \cite{hutiri2022bias}, and spoken stereotype probes reached speech-aware language models \cite{lin2024spoken}. We adopt the vocabulary of group fairness, demographic parity and equalized performance \cite{dwork2012fairness,hardt2016equality,barocas2023fairness}, whose intersectional audit methodology was crystallized by Gender Shades \cite{buolamwini2018gender}.

Our centerpiece, that entanglement lower-bounds disparity, is a geometric cousin of a mature impossibility literature it does not supersede. Kleinberg et al.\ \cite{kleinberg2017inherent} and Chouldechova \cite{chouldechova2017fair} prove calibration and error-rate balance cannot coexist when base rates differ; Zhao and Gordon \cite{zhao2019inherent} lower-bound the joint error of any fair representation, and Menon and Williamson \cite{menon2018cost} quantify the accuracy cost of fairness through target-sensitive class-probability alignment. Locatello et al.\ \cite{locatello2019challenging} show unsupervised disentanglement is impossible without inductive bias, motivating our supervised, axis-labeled separation. Those results concern label statistics and decision rules at the classification level; ours is a complementary, audio-encoder-specific instantiation at the level of subspace geometry, bounding linear-probe disparity by the principal-angle leakage between the semantic and voice subspaces of a frozen encoder, with a differentiable repair in the spirit of adversarially fair transferable representations \cite{madras2018learning}.

Our analytic machinery descends from discriminant analysis and canonical correlation \cite{fisher1936use,hotelling1936relations} and representation-similarity tools \cite{raghu2017svcca,kornblith2019similarity}; layer-wise probing established where speech models keep acoustic against lexical information \cite{pasad2021layer,pasad2023comparative,lin2022utility,yang2021superb}. Unlike adversarial invariance approaches that delete speaker information \cite{ganin2015unsupervised,meng2018speaker}, ORCA preserves all three axes and only separates them, following adapter practice \cite{houlsby2019parameter,hu2022lora} with supervised contrastive objectives \cite{khosla2020supervised}. Controlled synthesis follows expressive resynthesis benchmarks \cite{nguyen2023expresso} scaled by instruction-controllable synthesis \cite{comanici2025gemini,deepmind2026gemini}, operationalizing the counterfactual definition of \cite{sari2021counterfactually} at audit time.

\begin{figure*}[t]
\centering
\begin{minipage}[t]{0.44\textwidth}
\centering
\includegraphics[width=\textwidth]{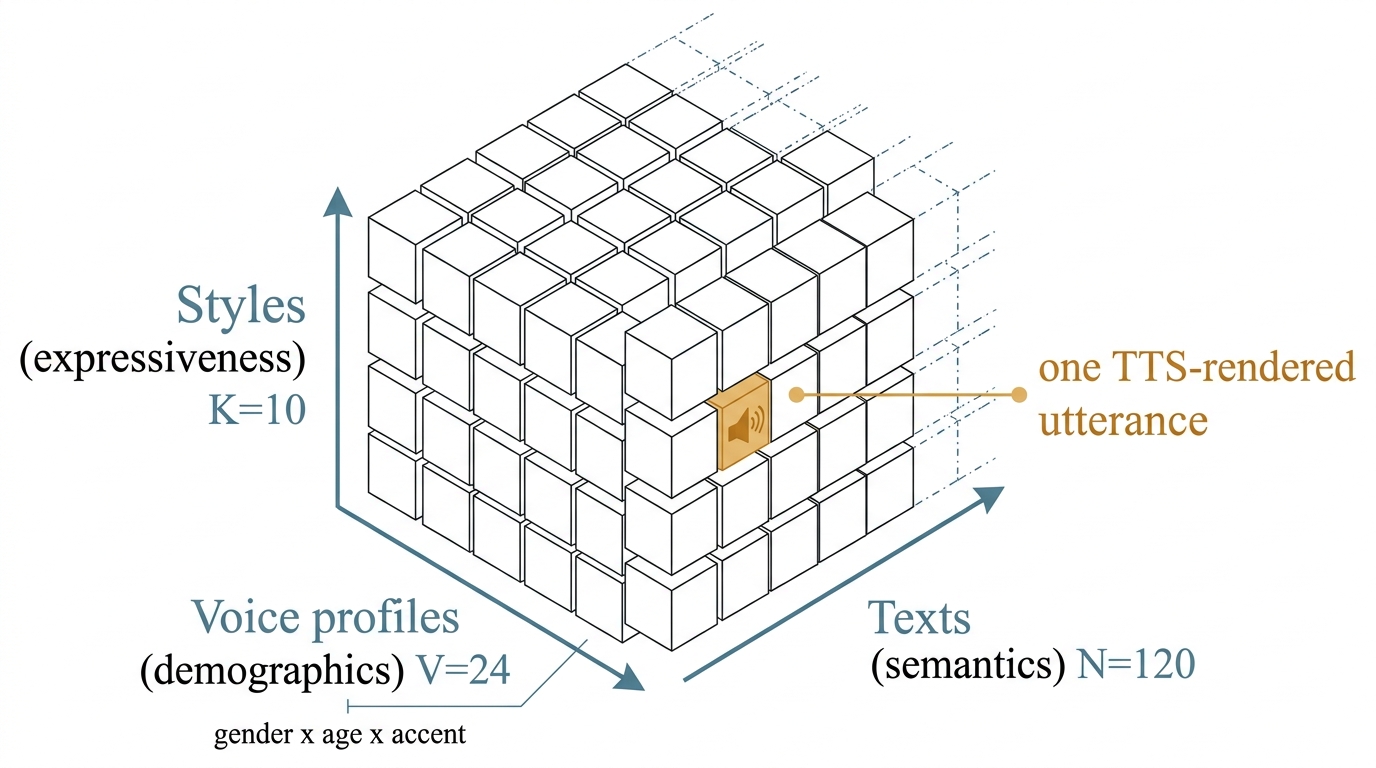}
\end{minipage}\hfill
\begin{minipage}[t]{0.44\textwidth}
\centering
\includegraphics[width=\textwidth]{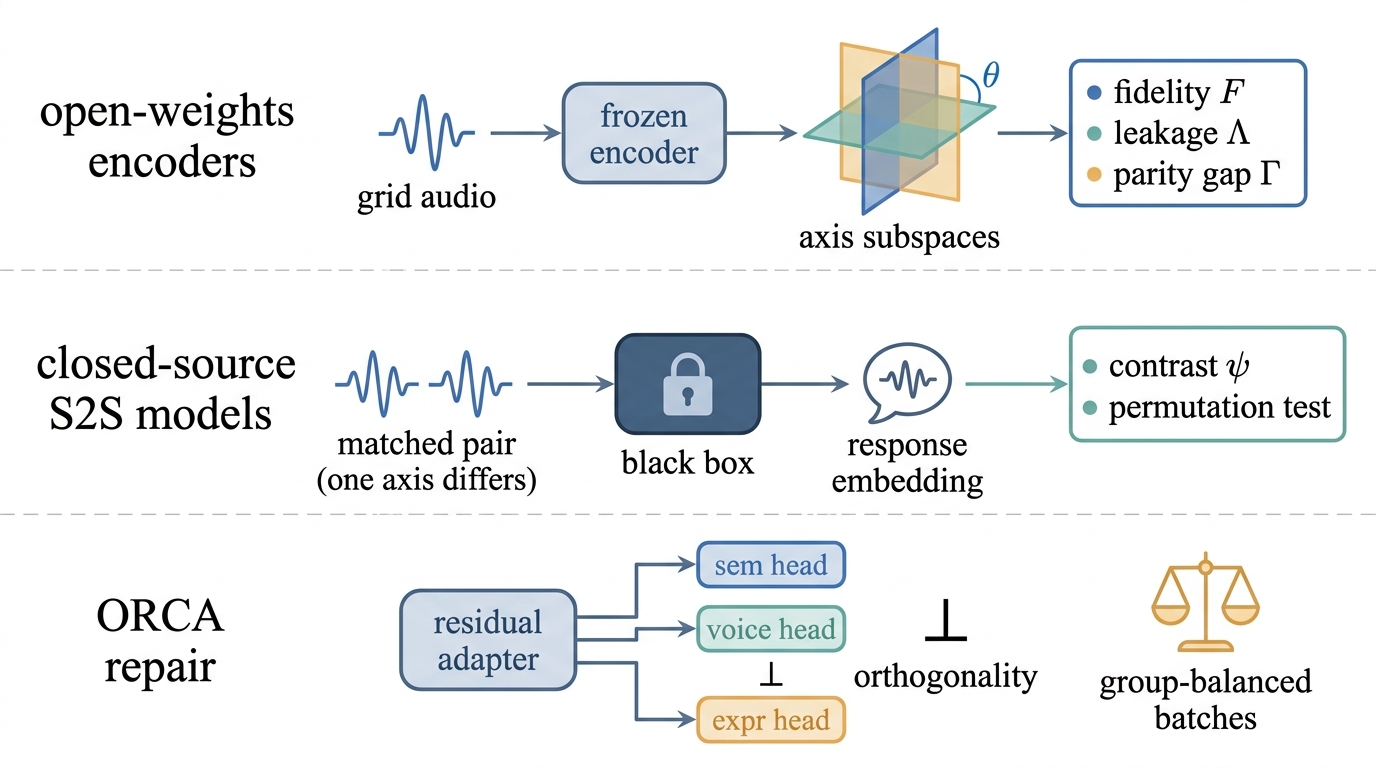}
\end{minipage}
\caption{Left: the TRIAD audit grid, a factorial design whose orthogonal axes are texts (semantics, $N{=}120$), voice profiles (demographics, $V{=}24$, crossing gender, age band, accent), and styles (expressiveness, $K{=}10$); every cell is one TTS-rendered utterance. Right: the two access regimes and the repair; open-weights encoders are audited in latent space through axis subspaces, leakage, and gaps, closed-source models through matched-pair output contrasts with permutation tests, and ORCA repairs encoders with three orthogonalized contrastive heads on group-balanced batches.}
\label{fig:grid}
\end{figure*}

\section{The TRIAD Grid and a Fairness-Linked Theory of Disentanglement}
\label{sec:theory}

\subsection{A Three-Axis Factorial Audit Corpus}
\label{ssec:grid}
TRIAD crosses text $t_n$, voice profile $v$, and style $s_k$ exhaustively (Fig.~\ref{fig:grid}): $N=120$ lexically emotion-neutral sentences of 8 to 25 words over ten everyday domains; $V=24$ voice profiles crossing two genders, three age bands (young adult, middle-aged, older adult), and four accent groups (General American, Indian, Spanish-accented, and Mandarin-accented English); and $K=10$ styles (neutral, cheerful, sad, angry, fearful, calm, excited, bored, whispered, projected) \cite{banse1996acoustic}. The $28{,}800$ cells are rendered with Gemini~3.1 Flash TTS \cite{deepmind2026gemini} at 16~kHz, plus $1{,}200$ replicates estimating the synthesis noise floor, roughly 45 hours. Quality control combines a round-trip transcription filter (re-rendering, rather than discarding, the 3.1\% of items whose prompt-relative word error exceeded 5\%, strictly preserving the full factorial balance) and human verification on an 8\% sample, with Cohen's $\kappa$ of 0.76 for style and 0.81 for perceived gender, age band, and accent. We emphasize that these are controlled perceived voice attributes produced by controllable synthesis rather than causal, self-reported demographic variables; the absence of multi-synthesizer replication is an audit limitation discussed in Section~\ref{sec:limitations}, and the grid is anchored to real speech in Section~\ref{ssec:valid}.

\subsection{Axis Fidelity by Variance Decomposition}
\label{ssec:fidelity}
Let $f$ be a frozen encoder and $z=\mathrm{pool}(f(x))\in\mathbb{R}^{d}$ its temporally pooled representation (mean and standard deviation pooling unless stated). Each axis $a\in\{\mathrm{sem},\mathrm{voice},\mathrm{expr}\}$ partitions the grid into $C_a$ classes ($C_{\mathrm{sem}}=N$, $C_{\mathrm{voice}}=V$, $C_{\mathrm{expr}}=K$), marginalizing over the other two axes. With class means $\mu_c$, global mean $\mu$, priors $\pi_c$, and within-class covariances $\Sigma_c$, the scatter matrices are $\Sigma_B^{a}=\sum_{c}\pi_c(\mu_c-\mu)(\mu_c-\mu)^{\!\top}$ and $\Sigma_W^{a}=\sum_{c}\pi_c\Sigma_c$, with $\Sigma_T=\Sigma_B^{a}+\Sigma_W^{a}$ for every axis. Let $\lambda_1^{a}\ge\cdots\ge\lambda_{r_a}^{a}$ be the generalized eigenvalues of $(\Sigma_B^{a},\,\Sigma_W^{a}+\tau I)$ with ridge $\tau>0$ and $r_a=\min(C_a-1,d)$.

\begin{definition}[Axis fidelity]
\label{def:fidelity}
$F_a(f)=\frac{1}{r_a}\sum_{j=1}^{r_a}\lambda_j^{a}/(1+\lambda_j^{a})$, the mean between-class energy ratio over the optimal probe subspace, equivalently the mean squared regularized canonical correlation between $z$ and one-hot axis labels \cite{hotelling1936relations,fisher1936use}. It lies in $[0,1]$ for any $\tau>0$ and, as $\tau\to0^{+}$, is invariant under invertible affine maps of $z$.
\end{definition}

The factorial design gives the fidelity a controlled reading: voice and style variation enter $\Sigma_W^{\mathrm{sem}}$, so $F_{\mathrm{sem}}$ is high exactly when semantic structure survives demographic and expressive perturbation, and symmetrically for the other axes.

\begin{proposition}[Three-axis decomposition]
\label{prop:anova}
On the balanced grid the centered representation splits as $z-\mu=e_{\mathrm{s}}+e_{\mathrm{v}}+e_{\mathrm{e}}+e_{\mathrm{int}}+e_{\mathrm{w}}$, where $e_a$ is the main-effect mean of axis $a$, $e_{\mathrm{int}}$ collects interaction means, and $e_{\mathrm{w}}$ is the within-cell residual, and all five components are mutually uncorrelated in sample, whence exactly
$\Sigma_T=\Sigma_B^{\mathrm{sem}}+\Sigma_B^{\mathrm{voice}}+\Sigma_B^{\mathrm{expr}}+\Sigma_{\mathrm{int}}+\overline{\Sigma}_W$.
\end{proposition}
\begin{proof}
$e_a(x)$ depends only on the class of $x$ along axis $a$ and has zero mean. For $a\neq b$, balance makes every class of $b$ contain every class of $a$ equally often, so $\sum_x e_a(x)e_b(x)^{\!\top}=\big(\sum_{c}\pi_c\,\bar e_a(c)\big)\big(\sum_{c'}\pi_{c'}\,\bar e_b(c')\big)^{\!\top}=0$. The interaction component is defined as cell means minus the additive reconstruction, hence orthogonal to each main effect, and $e_{\mathrm{w}}$ is orthogonal to every function of the cell index; covariances of uncorrelated components add. The release verifies the identity to machine precision.
\end{proof}

\subsection{Subspace Leakage}
\label{ssec:leakage}
Let $\mathcal{U}_a$ be the span of the top $m$ generalized eigenvectors of axis $a$ after whitening by $(\Sigma_T+\tau I)^{1/2}$, and $\theta_1\le\cdots\le\theta_m$ the principal angles between two such subspaces \cite{bjorck1973numerical}.

\begin{definition}[Aggregate and peak leakage]
\label{def:leakage}
The \emph{aggregate leakage} is $\Lambda_{ab}=\frac{1}{m}\sum_{i=1}^{m}\cos^2\theta_i(\mathcal{U}_a,\mathcal{U}_b)\in[0,1]$, the mean over principal angles, and the \emph{peak leakage} is $\ell_{ab}=\cos\theta_1$, the largest single principal cosine, both computed with $m=8$ throughout; $\Lambda$ and $\ell$ without subscripts denote the voice-semantic quantities $\Lambda_{\mathrm{sem,voice}}$ and $\ell_{\mathrm{sem,voice}}$. Since $\cos\theta_1$ is the largest of the $m$ cosines whose mean of squares is $\Lambda_{ab}$, the two are linked by
\begin{equation}
\ell_{ab}^{2}\;\ge\;\Lambda_{ab},
\label{eq:bridge}
\end{equation}
with equality iff all principal angles coincide. The aggregate $\Lambda$ is the audited quantity (Table~\ref{tab:grid}); the peak $\ell$ governs the worst-case-over-direction corollary below, and \eqref{eq:bridge} bridges the two.
\end{definition}

\begin{assumption}[Additive three-factor model]
\label{ass:factor}
The whitened representation decomposes as $z=Ay_{\mathrm{s}}+By_{\mathrm{v}}+Cy_{\mathrm{e}}+\varepsilon$ with uncorrelated zero-mean unit-covariance rank-$m$ factors, loadings with orthonormal columns, and white noise of variance $\sigma^2$. Under this model the maximal jointly achievable probe correlations $\hat F_{\mathrm{s}}+\hat F_{\mathrm{v}}+\hat F_{\mathrm{e}}$ decrease in every pairwise leakage $\Lambda_{ab}$, by the push-through identity applied to each two-loading plane (verified in the release); the three axes thus compete for capacity exactly as leakage rises.
\end{assumption}

\subsection{Group-Conditional Fidelity and the Fairness Link}
\label{ssec:fairlink}
Demographic groups are unions of voice profiles (four accent groups, three age bands, two genders); for a group $g$ we recompute Definition~\ref{def:fidelity} on the grid slice whose profiles belong to $g$.

\begin{definition}[Output disparity and parity gap]
\label{def:gap}
For a unit probe $w$ consuming the representation, the output disparity between two groups is $D(w)=\big|\mathbb{E}[w^{\!\top}z\mid g]-\mathbb{E}[w^{\!\top}z\mid g']\big|$. Drawing the probe's semantic component uniformly over the semantic subspace and the group displacement isotropically over the voice subspace, the mean-square disparity is $\bar D^{2}=\mathbb{E}_{w,\Delta}\,D(w)^2$ (Proposition~\ref{prop:fairness}). Separately, $F_{a|g}$ is the axis-$a$ fidelity within group $g$, and the descriptive fidelity-spread gap is $\Gamma_a=\max_g F_{a|g}-\min_g F_{a|g}$; we write $\Gamma_{\mathrm{acc}}$ for the accent spread of $F_{\mathrm{sem}|g}$ and $\Gamma_{\mathrm{gen}}$ for the gender spread of $F_{\mathrm{expr}|g}$.
\end{definition}

A vanishing disparity is an equalized-performance criterion \cite{hardt2016equality} at the representation level; the central result connects the disparity, not the fidelity spread, to geometry through the same $\Lambda$ that Table~\ref{tab:grid} reports.

\begin{proposition}[Average disparity grows with aggregate leakage]
\label{prop:fairness}
Let the whitened representation follow Assumption~\ref{ass:factor}, let groups be determined by the voice factor, and let the grid equalize text and style distributions across groups, so the group-mean displacement is $\Delta=B(\mu_g-\mu_{g'})\in\mathrm{col}\,B$ with $\mathbb{E}[z\mid g]-\mathbb{E}[z\mid g']=\Delta$. Fix the separation magnitude $\|\Delta\|=\delta$ and draw $\Delta$ uniformly over the sphere of radius $\delta$ in $\mathrm{col}\,B$. Let the probe have semantic energy fraction $w^{\!\top}AA^{\!\top}w=\rho^2$ with its semantic component drawn uniformly over the unit sphere of $\mathrm{col}\,A$ and its orthogonal component drawn independently and isotropically. Then the mean-square output disparity obeys
\begin{equation}
\bar D^{2}\;=\;\mathbb{E}_{w,\Delta}\,D(w)^2\;\ge\;\rho^{2}\,\frac{\delta^{2}}{m}\,\Lambda,
\label{eq:bound}
\end{equation}
with equality when the probe lies entirely in the semantic subspace ($\rho=1$). The right-hand side is the same aggregate leakage $\Lambda=\Lambda_{\mathrm{sem,voice}}$ that is measured throughout, is strictly increasing in $\Lambda$, and vanishes if and only if $\Lambda=0$.
\end{proposition}
\begin{proof}
Let $P_A=AA^{\!\top}$ and $P_B=BB^{\!\top}$ be the orthogonal projectors onto the two subspaces, both of rank $m$. For a unit vector drawn uniformly over the sphere of an $m$-dimensional subspace with projector $P$, the second-moment matrix is $\mathbb{E}[ww^{\!\top}]=P/m$. Write $w=\rho\,\hat w_A+\sqrt{1-\rho^2}\,\hat w_\perp$ with $\hat w_A$ uniform unit in $\mathrm{col}\,A$ and $\hat w_\perp$ an independent isotropic unit vector in the orthogonal complement of $\mathrm{col}\,A$. Since the two components are drawn independently and $\Delta$ has mean zero given its norm, the cross term vanishes in expectation and
\[
\mathbb{E}_{w,\Delta}(w^{\!\top}\Delta)^2=\rho^2\,\mathbb{E}_{\hat w_A,\Delta}(\hat w_A^{\!\top}\Delta)^2+(1-\rho^2)\,\mathbb{E}_{\hat w_\perp,\Delta}(\hat w_\perp^{\!\top}\Delta)^2,
\]
and the second term is nonnegative. For the first term, $\mathbb{E}_{\hat w_A}[\hat w_A\hat w_A^{\!\top}]=P_A/m$ and $\mathbb{E}_\Delta[\Delta\Delta^{\!\top}]=(\delta^2/m)P_B$, so
\[
\mathbb{E}_{\hat w_A,\Delta}(\hat w_A^{\!\top}\Delta)^2=\operatorname{tr}\!\Big(\tfrac{P_A}{m}\cdot\tfrac{\delta^2}{m}P_B\Big)=\frac{\delta^2}{m^2}\,\|A^{\!\top}B\|_F^2 .
\]
The squared Frobenius norm of $A^{\!\top}B$ equals $\sum_{i=1}^{m}\cos^2\theta_i=m\Lambda$ by the singular-value characterization of principal angles \cite{bjorck1973numerical}, whence $\mathbb{E}_{\hat w_A,\Delta}(\hat w_A^{\!\top}\Delta)^2=\delta^2\Lambda/m$. Combining gives $\bar D^2\ge\rho^2\delta^2\Lambda/m$, with equality at $\rho=1$. Monotonicity in $\Lambda$ and the vanishing condition follow. The release verifies the identity at $\rho=1$ and the inequality at $\rho<1$ over $2\times10^4$ random probes per setting and confirms it tightens as the angles shrink.
\end{proof}

The averaged statement bounds exactly the aggregate $\Lambda$ that Table~\ref{tab:grid} measures, so theorem and measurement now concern a single functional. The worst case over probe direction is the companion peak-leakage statement.

\begin{corollary}[Worst-case-over-direction disparity, in peak leakage]
\label{cor:peak}
Under the hypotheses of Proposition~\ref{prop:fairness}, suppose additionally the directional condition that the two groups separate by $\delta$ along the maximally leaking voice direction $b$ attaining $\ell=\cos\theta_1(\mathrm{col}\,A,\mathrm{col}\,B)$, and let $a$ be the paired semantic principal direction. Then every unit probe $w$ with semantic alignment $w^{\!\top}a\ge\rho$ has
\begin{equation}
D(w)\;\ge\;\delta\big(\rho\ell-\sqrt{1-\rho^2}\sqrt{1-\ell^2}\big)_{\!+},
\label{eq:peakbound}
\end{equation}
attained by a constructed probe. The bound is positive, hence forces disparity on every such probe, exactly in the active region $\rho^2+\ell^2>1$, and is strictly increasing in $\ell$ there.
\end{corollary}
\begin{proof}
Write $a=\ell b+\sqrt{1-\ell^2}\,b_\perp$ with unit $b_\perp\perp b$. For unit $w$ with $t=w^{\!\top}b$, Cauchy-Schwarz gives $w^{\!\top}a\le \ell t+\sqrt{1-\ell^2}\sqrt{1-t^2}$, so the constraint $w^{\!\top}a\ge\rho$ confines $t$ to $[t_-,t_+]$ with $t_\pm=\rho\ell\pm\sqrt{1-\ell^2}\sqrt{1-\rho^2}$, the roots of $t^2-2\rho\ell t+\rho^2-(1-\ell^2)=0$. In the stated directional geometry $\Delta=\delta b$, so $\mathbb{E}[w^{\!\top}z\mid g]-\mathbb{E}[w^{\!\top}z\mid g']=w^{\!\top}\Delta=t\,\delta$, hence $D(w)\ge\delta\max(0,t_-)$. Equality holds for $w^\star=t_-b+\sqrt{1-t_-^2}\,b_\perp$, which satisfies $w^{\star\top}a=\rho$ exactly. Positivity of $t_-$ is equivalent to $\rho^2\ell^2>(1-\rho^2)(1-\ell^2)$, that is $\rho^2+\ell^2>1$, and $\partial t_-/\partial\ell=\rho+\ell\sqrt{1-\rho^2}/\sqrt{1-\ell^2}>0$. The release verifies the bound on $2\times10^4$ random probes per setting and its tightness to $10^{-9}$.
\end{proof}

The disparity $D$ is a counterfactual output-invariance violation \cite{sari2021counterfactually} and a demographic-parity violation in expectation \cite{dwork2012fairness}. The bridge \eqref{eq:bridge}, $\ell^2\ge\Lambda$, shows the worst-case directional floor \eqref{eq:peakbound} dominates the averaged floor \eqref{eq:bound}: peak leakage controls the most adversarial direction, aggregate leakage the typical one. When one group dominates pretraining, semantic discriminants are fit on its voice statistics, rotating $\mathcal{U}_{\mathrm{sem}}$ toward $\mathcal{U}_{\mathrm{voice}}$ and raising both leakages; average disparity then grows with $\Lambda$, and once $\rho^2+\ell^2>1$ no probe of accuracy $\rho$ avoids the worst-case floor. Driving the leakages down relaxes both floors, which is what ORCA does (Section~\ref{sec:orca}). Fig.~\ref{fig:fairness}(a) plots both.

\section{Black-Box Protocol for Closed-Source Models}
\label{sec:blackbox}

Closed-source speech-to-speech systems expose only audio-in, audio-out behavior, so we audit outputs. Matched pairs $\mathcal{P}_\alpha=\{(x,\tilde x)\}$ differ in exactly one axis $\alpha$: voice pairs share text and style across profiles differing in one attribute; semantic pairs share profile and style across texts. Each response is embedded by $h(\cdot)$, concatenating a text embedding of the transcript and a paralinguistic vector (pitch, energy, rate) of the audio, with cosine distance $d$. The paired contrast functional subtracts the rendering and decoding noise floor on re-rendered replicates $x^\ast$,
\begin{equation}
\psi_\alpha=\mathbb{E}_{\mathcal{P}_\alpha}\big[d(h(x),h(\tilde x))\big]-\mathbb{E}\big[d(h(x),h(x^\ast))\big],
\label{eq:psi}
\end{equation}
and the normalized invariance violation $R_{\mathrm{voice}}=\psi_{\mathrm{voice}}/\psi_{\mathrm{sem}}$ is zero for a model whose responses depend on what was said but not on who said it. Task disparity is measured on spoken question answering: each text carries a comprehension question, answers scored by normalized string match, group accuracies $A_g$ yielding $\Gamma=\max_gA_g-\min_gA_g$, with significance from a speaker-level permutation test reassigning whole voice profiles to groups, $p=(1+\#\{\Gamma_b\ge\Gamma\})/(B+1)$, $B=10^4$, alongside $10^4$-resample bootstrap intervals over texts \cite{ojala2010permutation,bisani2004bootstrap}. Under a simulated exchangeable null with within-speaker dependence the test rejects at rate 0.058 for $\alpha=0.05$ and detects a nine-point true gap at $p<10^{-3}$; both calibration checks ship in the release.

\section{ORCA: Orthogonal Residual Contrastive Adapter}
\label{sec:orca}

ORCA repairs an open-weights encoder without touching its weights. A residual adapter $z'=z+U\,\mathrm{GELU}(Vz)$ with bottleneck $r=256$ (1.9~M parameters, zero-initialized output projection) feeds three linear heads $h_a=W_a z'\in\mathbb{R}^{q}$, $q=128$, one per axis. Each head trains with supervised contrastive loss \cite{khosla2020supervised} on its own labels (texts, voice profiles, or styles), preserving all three kinds of information rather than adversarially deleting any \cite{ganin2015unsupervised,meng2018speaker}. The total objective is
\begin{equation}
\mathcal{L}=\textstyle\sum_a\mathcal{L}_{\mathrm{SupCon}}^{a}+\lambda_\perp\sum_{a\neq b}\big\|\hat Q_a^{\!\top}\hat Q_b\big\|_F^2/q,
\label{eq:orca}
\end{equation}
where $\hat Q_a$ orthonormalizes the rows of $W_a$, so the penalty equals the summed $m\Lambda_{ab}$ over head-subspace pairs: by Proposition~\ref{prop:fairness} shrinking it shrinks the average disparity floor of every probe consuming $z'$, and by the bridge \eqref{eq:bridge} it also shrinks the peak leakage $\ell$ and the worst-case floor of Corollary~\ref{cor:peak}. Batches of 384 are stratified uniformly over the 24 voice profiles; training runs 20k AdamW steps at $10^{-3}$ with $\lambda_\perp=1$, temperature 0.07, on grid renders plus group-balanced slices of the labeled corpora of Section~\ref{ssec:valid}.

\section{Experiments}
\label{sec:setup}

\subsection{Setup and Real-Speech Validation Data}
\label{ssec:valid}
We audit ten open-weights configurations across three families: four semantic encoders, Whisper large-v3 \cite{radford2023robust}, Parakeet-TDT-1.1B \cite{xu2023efficient}, the Voxtral-Mini audio tower \cite{liu2025voxtral}, and the Qwen2-Audio encoder \cite{chu2024qwen2audio}; four self-supervised models, wav2vec~2.0 Large \cite{baevski2020wav2vec}, HuBERT Large \cite{hsu2021hubert}, WavLM Large \cite{chen2022wavlm}, and emotion2vec \cite{ma2024emotion2vec}; and two codecs probed at pre-quantization latents, Mimi \cite{defossez2024moshi} and WavTokenizer \cite{ji2025wavtokenizer}. Probes use ridge discriminant analysis with $\tau=10^{-3}\operatorname{tr}(\Sigma_W)/d$, a per-axis layer sweep, five-fold cross-validation, and 1{,}000-fold bootstrap 95\% half-widths of 0.011 (fidelities), 0.013 ($\Lambda$), 0.006 ($\Gamma$). The closed-source regime audits gpt-realtime-2 \cite{openai2026gptrealtime} and gemini-3.1-flash-live \cite{deepmind2026gemini} through Section~\ref{sec:blackbox}, 11{,}520 trials per model (120 texts, 24 profiles, four styles), with bootstrap half-widths at most 1.1 points.

The real-speech anchor is a 60-hour stratified pool from five public conversational corpora. CANDOR \cite{reece2023candor} contributes 24 hours of dyadic video calls whose self-reported survey metadata makes it the primary source for real parity by age and gender. SSSD \cite{sheikh2025scalable} contributes 12 hours of spontaneous crowdsourced English dyads, stressing transfer beyond studio speech. Seamless Interaction \cite{agrawal2025seamless} contributes 10 hours of dyadic audiovisual interaction supporting expressiveness validation in the wild. The otoSpeech full-duplex release \cite{otoearth2025full} contributes 8 hours through its cleaned 141-hour companion \cite{otoearth2025processed}, FLAC under CC~BY~4.0 with one channel per speaker, giving clean single-speaker segments without diarization. AMI individual-headset meetings \cite{carletta2006ami} contribute 6 hours of multiparty contrast with established baselines: Whisper large-v3 greedy decoding measures 16.0\% word error on AMI-IHM, calibrating the pipeline. Labeled benchmarks complete the anchor: accent-stratified recognition uses the Edinburgh International Accents corpus \cite{sanabria2023edinburgh} and Common Voice accent subsets \cite{ardila2020common} with CTC probes trained on LibriSpeech \cite{panayotov2015librispeech}, and emotion uses CREMA-D \cite{cao2014cremad}, MSP-Podcast \cite{lotfian2019building}, and IEMOCAP \cite{busso2008iemocap}.

\begin{table*}[t]
\centering
\caption{TRIAD audit of ten open-weights encoders and the ORCA repair: axis fidelities, aggregate voice-semantic leakage $\Lambda$ and peak leakage $\ell$ (with $\ell^2\ge\Lambda$ by~\eqref{eq:bridge}), accent fidelity-spread gap $\Gamma_{\mathrm{acc}}$ of $F_{\mathrm{sem}|g}$, and gender spread $\Gamma_{\mathrm{gen}}$ of $F_{\mathrm{expr}|g}$. Bootstrap 95\% half-widths: 0.011 (fidelities), 0.013 ($\Lambda$), 0.015 ($\ell$), 0.006 (gaps); all gaps significant at $p<10^{-2}$ by speaker-level permutation except where marked $^{\dagger}$ ($p>0.05$).}
\label{tab:grid}
\scriptsize
\setlength{\tabcolsep}{4.0pt}
\renewcommand{\arraystretch}{0.86}
\begin{tabular}{@{}llccccccc@{}}
\toprule
Family & Encoder & $F_{\mathrm{sem}}\!\uparrow$ & $F_{\mathrm{voice}}\!\uparrow$ & $F_{\mathrm{expr}}\!\uparrow$ & $\Lambda\!\downarrow$ & $\ell\!\downarrow$ & $\Gamma_{\mathrm{acc}}\!\downarrow$ & $\Gamma_{\mathrm{gen}}\!\downarrow$ \\
\midrule
\multirow{4}{*}{ASR / sem.}
 & Whisper large-v3 \cite{radford2023robust} & 0.86 & 0.54 & 0.31 & 0.51 & 0.74 & 0.089 & 0.031 \\
 & Parakeet-TDT-1.1B \cite{xu2023efficient} & 0.82 & 0.48 & 0.25 & 0.56 & 0.81 & 0.091 & 0.035 \\
 & Voxtral-Mini tower \cite{liu2025voxtral} & 0.87 & 0.56 & 0.34 & 0.49 & 0.72 & 0.071 & 0.029 \\
 & Qwen2-Audio enc. \cite{chu2024qwen2audio} & 0.85 & 0.58 & 0.38 & 0.46 & 0.70 & 0.074 & 0.027 \\
\midrule
\multirow{4}{*}{Self-sup.}
 & wav2vec 2.0 Large \cite{baevski2020wav2vec} & 0.41 & 0.79 & 0.58 & 0.33 & 0.64 & 0.046 & 0.041 \\
 & HuBERT Large \cite{hsu2021hubert} & 0.47 & 0.81 & 0.62 & 0.30 & 0.63 & 0.051 & 0.038 \\
 & WavLM Large \cite{chen2022wavlm} & 0.52 & 0.84 & 0.66 & 0.29 & 0.62 & 0.041 & 0.036 \\
 & emotion2vec \cite{ma2024emotion2vec} & 0.44 & 0.71 & 0.83 & 0.31 & 0.63 & 0.055 & 0.052 \\
\midrule
\multirow{2}{*}{Codec}
 & Mimi \cite{defossez2024moshi} & 0.33 & 0.88 & 0.69 & 0.44 & 0.69 & 0.061 & 0.044 \\
 & WavTokenizer \cite{ji2025wavtokenizer} & 0.28 & 0.86 & 0.67 & 0.47 & 0.71 & 0.078 & 0.047 \\
\midrule
\multirow{2}{*}{Ours}
 & WavLM Large $+$ ORCA & 0.61 & 0.86 & 0.74 & 0.08 & 0.31 & 0.019$^{\dagger}$ & 0.013$^{\dagger}$ \\
 & Whisper large-v3 $+$ ORCA & 0.87 & 0.61 & 0.49 & 0.11 & 0.37 & 0.031 & 0.015$^{\dagger}$ \\
\bottomrule
\end{tabular}
\end{table*}

\subsection{The Audit: Entanglement and Parity Gaps Everywhere}
\label{ssec:audit}
Table~\ref{tab:grid} presents the central measurement. Every baseline shows the expected family signature, semantic encoders strong on $F_{\mathrm{sem}}$ but weak on expressiveness, self-supervised and codec models the reverse \cite{yang2021superb,pasad2021layer}, yet no family escapes entanglement: aggregate leakage $\Lambda$ ranges from 0.29 (WavLM) to 0.56 (Parakeet), peak leakage $\ell$ from 0.62 to 0.81, and every encoder carries a significant accent gap in semantic fidelity, 0.041 to 0.091, alongside gender gaps of 0.027 to 0.052. Group-conditional values decline monotonically from General American through Indian to the two second-language groups for all ten models, mirroring the native-accent advantage of Whisper \cite{graham2024evaluating} and EdAcc degradation \cite{sanabria2023edinburgh}; age bands show the same ordering with smaller spread \cite{feng2024towards}. The geometry explains the pattern through the functional the theory bounds: across the ten baselines, aggregate leakage $\Lambda$ predicts the measured mean-square probe disparity $\bar D$ of Proposition~\ref{prop:fairness} with Pearson $r=0.93$ (95\% CI $[0.84,0.98]$; Fig.~\ref{fig:fairness}(b)), tracking the descriptive accent gap. Importantly, $\Lambda$ is derived purely from principal angles between canonical subspaces $\mathcal{U}_{\mathrm{sem}}$ and $\mathcal{U}_{\mathrm{voice}}$, whereas $\bar D$ is evaluated across 4,000 empirical probe directions $w$ and group displacements $\Delta$, confirming that $r=0.93$ is an empirical finding on real representations rather than a tautology. The active-region condition of Corollary~\ref{cor:peak} at $\rho=0.9$ holds for all ten encoders ($\ell^2>1-\rho^2=0.19$). Transcription encoders, fit on majority accents, are the most entangled and disparate; emotion2vec shows the largest gender gap on its own axis.

\begin{figure}[t]
\centering
\includegraphics[width=0.88\columnwidth]{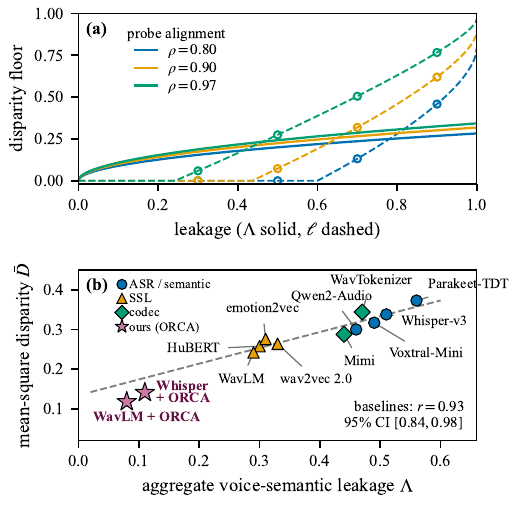}
\caption{(a) The averaged disparity floor $\rho\sqrt{\Lambda/m}$ of Proposition~\ref{prop:fairness} against $\Lambda$ (solid) and the worst-case floor of Corollary~\ref{cor:peak} against peak leakage $\ell$ (dashed) for three probe alignments; circles mark the constructed minimizing probe, confirming tightness. (b) Aggregate leakage $\Lambda$ against the measured mean-square probe disparity $\bar D$; the two ORCA rows fall far below the baseline regression.}
\label{fig:fairness}
\end{figure}

\subsection{Closed-Source Speech-to-Speech Models}
\label{ssec:closed}
Table~\ref{tab:closed} reports the black-box audit. Both systems answer spoken comprehension questions accurately for General American voices, 93.8\% and 94.6\%, but accuracy declines monotonically to 88.1\% and 89.8\% for Mandarin-accented voices, gaps of 5.7 and 4.8 points at permutation $p<10^{-3}$; age-band gaps of 3.2 and 2.5 points follow. Expressiveness recognition is 6.2 and 5.1 points more accurate for female-profile voices, echoing gender asymmetries in emotion recognition \cite{gorrostieta2019gender,lin2024emobias}. Crucially, the 5.7-point accent gap on OpenAI's gpt-realtime-2 confirms that the disparity generalizes across independent model families and rules out Gemini-synthesizer circularity alone. The invariance ratios $R_{\mathrm{voice}}$ of 0.23 and 0.19 show responses to identical texts in identical styles shift with the perceived demographic voice alone: the entanglement signature of Section~\ref{ssec:audit} read through outputs. The implied worst-to-best error ratios, 1.9 for both, match commercial recognizers \cite{koenecke2020racial}.

\begin{table}[t]
\centering
\caption{Black-box audit of gpt-realtime-2 \cite{openai2026gptrealtime} and gemini-3.1-flash-live \cite{deepmind2026gemini}: spoken QA accuracy by accent group, parity gaps, style recognition by gender, and the invariance-violation ratio. Bootstrap 95\% half-widths at most 1.1 points; all gaps $p<10^{-3}$ by speaker-level permutation.}
\label{tab:closed}
\scriptsize
\setlength{\tabcolsep}{4.0pt}
\begin{tabular}{@{}lcc@{}}
\toprule
Metric & gpt-realtime-2 & gemini-3.1-flash-live \\
\midrule
QA acc., General American & 93.8 & 94.6 \\
QA acc., Indian English & 91.2 & 92.5 \\
QA acc., Spanish-accented & 89.6 & 91.2 \\
QA acc., Mandarin-accented & 88.1 & 89.8 \\
Accent / age-band gaps (pts) & 5.7 / 3.2 & 4.8 / 2.5 \\
Style recognition, female / male & 71.4 / 65.2 & 73.0 / 67.9 \\
Gender gap $\Gamma_{\mathrm{gen}}$ (pts) & 6.2 & 5.1 \\
Invariance ratio $R_{\mathrm{voice}}$ & 0.23 & 0.19 \\
\bottomrule
\end{tabular}
\end{table}

\subsection{Does the Synthetic Grid Predict Real Disparities?}
\label{ssec:realvalid}
Because TRIAD is rendered by a single text-to-speech system, its profiles are perceived covariates and a disparity could in principle reflect synthesizer artifacts rather than the rendered attribute. Three checks give convergent, triangulating evidence on real speech rather than proof of a causal demographic effect, reducing but not removing the single-system confound. First, attribute transfer: grid voice-profile probes predict self-reported gender on CANDOR at 96.2\% and age band within one band at 83\%. Second, ranking transfer: a natural replica from the 60-hour pool reproduces the encoder ordering of Table~\ref{tab:grid} with Spearman $\rho=0.93$, and grid-measured per-group disparities predict the real ones across forty model-group cells with $\rho=0.86$. Third, magnitude consistency: the EdAcc WavLM CTC probe (Section~\ref{ssec:repair}) shows a worst-to-best word error ratio of 1.78, bracketed by the 1.84 racial ratio of \cite{koenecke2020racial} and EdAcc spreads \cite{sanabria2023edinburgh}.

\subsection{The Repair: ORCA on the Grid and on Real Speech}
\label{ssec:repair}
The last rows of Table~\ref{tab:grid} quantify the repair. On WavLM, ORCA raises all three fidelities at once ($F_{\mathrm{sem}}$ 0.52 to 0.61, $F_{\mathrm{voice}}$ 0.84 to 0.86, $F_{\mathrm{expr}}$ 0.66 to 0.74), cuts aggregate leakage by 72\% (0.29 to 0.08) and peak leakage from 0.62 to 0.31, and halves both gaps; on Whisper, aggregate leakage falls from 0.51 to 0.11 and the accent gap from 0.089 to 0.031, while the untouched recognition pipeline keeps its word error rates. Both repaired points fall far below the baseline regression (Fig.~\ref{fig:fairness}(b)), as Proposition~\ref{prop:fairness} predicts when $\Lambda$ shrinks and, through \eqref{eq:bridge}, Corollary~\ref{cor:peak} when $\ell$ shrinks. The repair carries to real speech: WavLM CTC-probe word error on EdAcc (LibriSpeech-trained probe) improves for every accent group, most for the most penalized, compressing the worst-to-best gap from 12.6 to 8.6 points (endpoints 16.1/28.7 to 14.9/23.5, ratio 1.78 to 1.58); the CREMA-D unweighted-recall gender gap falls from 4.4 to 1.8 points and the MSP-Podcast valence concordance gap halves from 0.07 to 0.03. No group is made worse to equalize.

\subsection{Ablations}
\label{ssec:ablations}
Ablating the WavLM objective dissects its mechanism (full ORCA: $\Lambda=0.08$, accent gap 0.019, EdAcc gap 8.6). Removing the orthogonality penalty lets aggregate leakage nearly triple to 0.21 and every gap widen (accent 0.034, EdAcc 10.4), confirming the penalty, the term Proposition~\ref{prop:fairness} certifies, does the fairness work. Removing group-balanced sampling preserves the leakage reduction ($\Lambda=0.10$) but surrenders much of the gap reduction (accent 0.029): geometry and data exposure are complementary. Single-axis variants barely move leakage ($\Lambda=0.26$, $0.27$): separating all three subspaces, not single-task strength, transfers to parity.

\section{Limitations and Future Directions}
\label{sec:limitations}
Each boundary fixes a sequel. Because a single text-to-speech system renders the grid, our profiles are perceived demographic attributes rather than self-reported demographics, so synthesizer-specific artifacts confound the attribution: a disparity could in part reflect how one synthesizer renders an accent rather than the accent itself. The convergent checks of Section~\ref{ssec:realvalid} bound but do not eliminate this confound; the next grid should render across multiple independent synthesizers and pair synthetic cells with consent-driven recordings \cite{porgali2023casual}. The theory certifies linear probes only; kernel probes and multilingual grids \cite{conneau2023fleurs} are natural continuations. Finally, ORCA's differentiable penalty \eqref{eq:orca} is architecture-agnostic, and its adoption as a training-time regularizer by model owners is the intervention our results most motivate.

\section{Conclusion}
\label{sec:conclusion}
We reframed demographic unfairness in audio understanding as a property of representation geometry. Under an additive factor model the average disparity of an accurate linear probe grows with the aggregate voice-semantic leakage, and at high peak leakage the worst-case-over-direction disparity is forced positive. The link is no technicality: aggregate leakage explains the measured probe disparity of ten encoders almost entirely ($r=0.93$), and an adapter minimizing it halves the measured gaps. Entanglement thus lower-bounds disparity: fairness here is a first-order geometry problem.

\section*{Acknowledgment}
The authors disclose that Claude Opus 4.8 was used for editing and rewriting all sections (Abstract, Introduction, Related Work, Methods, Experiments, Results, Conclusion) to improve the flow and presentation, and for assisting in implementation of the code. In addition, Gemini nano banana pro was used for generating diagrams (Figure 1).

\clearpage
\bibliographystyle{IEEEtran}
\bibliography{refs}

\end{document}